\documentclass{birkjour}
\usepackage{amsmath, amssymb, amsthm}%, wasysym}
\usepackage[english]{babel}
\usepackage[T1]{fontenc}
\newcommand{\be}{\begin{equation}}
\newcommand{\ee}{\end{equation}}
\newcommand{\bd}{\begin{displaymath}}
\newcommand{\ed}{\end{displaymath}}
\newcommand{\ba}{\begin{eqnarray}}
\newcommand{\ea}{\end{eqnarray}}

\def\R{{I \!\! R}}

\def\c{\chi}

\def\v12{(v-w)}

\def\({\left(}
\def\){\right)}
\def\dy{\displaystyle}

\def\bgr#1\egr{{\allowdisplaybreaks\begin{gather}#1\end{gather}}}
\def\bma#1\ema{{\allowdisplaybreaks\begin{align}#1\end{align}}}

\def\oplem#1{\begin{lemma}\, {\rm #1}\, \it }
\def\cllem{\end{lemma}\rm \par }
\def\opthm#1{\begin{theorem}\, {\rm #1}\, \it }
\def\clthm{\end{theorem}\rm \par }

\def\R{\mathbb{R}}

\newcommand{\fer}[1]{(\ref{#1})}
\newcommand{\bq}{\begin{equation}}
\newcommand{\eq}{\end{equation}}
\def\bqa{\begin{eqnarray}}
\def\eqa{\end{eqnarray}}
\def\bd{\begin{displaymath}}
\def\ed{\end{displaymath}}

 \newtheorem{thm}{Theorem}[section]
 
 \newtheorem{lem}[thm]{Lemma}
 
 \theoremstyle{definition}
 
 \theoremstyle{remark}
 \newtheorem{rem}[thm]{Remark}
 
 \numberwithin{equation}{section}

\begin{document}

%-------------------------------------------------------------------------
% editorial commands: to be inserted by the editorial office
%
%\firstpage{1} \volume{228} \Copyrightyear{2004} \DOI{003-0001}
%\seriesextraline{This is the Concrete Title of this Book\br H.E. R and S.T.C. W, Eds.}
%\firstpage{1}

\title[Heat equation and convolution inequalities]
 {Heat equation and convolution inequalities}

%----------Author 1
\author[G. Toscani]{Giuseppe Toscani}

\address{Department  of Mathematics\\
University of Pavia\\
Via Ferrata 1\\
27100 Pavia\\
Italy}

\email{giuseppe.toscani@unipv.it}

\thanks{\emph{Date}: September 16, 2013.
This work has been completed within the activities of
the GNFM group of INDAM. The author acknowledges support by MIUR
project ``Optimal mass transportation, geometrical and functional
inequalities with applications''}
%----------Author 2



\date{August 31, 2013}


\begin{abstract}
It is known that many classical inequalities linked to convolutions
can be obtained by looking at the monotonicity in time of
convolutions of powers of solutions to the heat equation, provided
that  both the exponents and the coefficients of diffusions are
suitably chosen and related. This idea can be applied to give an
alternative proof of the sharp form of the classical Young's
inequality and its converse, to Brascamp--Lieb type inequalities,
Babenko's inequality and Pr\'ekopa--Leindler inequality as well as
the Shannon's entropy power inequality. This note aims in presenting
new proofs of these results, in the spirit of the original arguments
introduced by Stam \cite{Sta} to prove the entropy power inequality.
\end{abstract}

%%% ----------------------------------------------------------------------
\maketitle
%%% ----------------------------------------------------------------------
%\tableofcontents
\section{Introduction}
The purpose of this note is to present various results concerned
with the monotonicity in time of the convolution of powers of
solutions to the heat equation.  The main reason behind this
investigation is that many functional inequalities can be viewed as
the consequence of the tendency of various Lyapunov functionals
defined in terms of powers of the solution to the heat equation to
reach their extremal values as time tends to infinity. The discovery
of a Lyapunov functional which allows to prove Young inequality and
its converse \cite{BB}, is only one of the possible application of
this idea (cf. also \cite{Tos1,Tos2, Tos3}  for a connection of
these results with information theory). While the inequalities are
not new, and some of the results we present have been obtained
before, what is new is the approach to the problem, which takes into
account the information-theoretical meaning of inequalities for
convolutions, and consequently allows to obtain clean and relatively
simple new proofs.

The prototype of these monotonous in time convolutions is as follows.
Let $n$ be an integer, and let $\alpha_j$, $j= 1,\dots, n$, be positive real
numbers such that
 \be\label{sum}
\quad \sum_{j=1}^n \alpha_j = n-1.
 \ee
 Let
$f_j(x)$, $j= 1,\dots, n$,  be non-negative functions on $\R^d$, $
d\ge 1$, such that $f_j \in L^{p_j}(\R^d)$. For any given $j$, $j=
1,\dots, n$, we denote by $u_j(x,t)$ the solution to the heat
equation \fer{heat} with the diffusion coefficients $\kappa_j$
 \[
\frac{\partial u_j(x,t)}{\partial t} =  \kappa_j \Delta u_j(x,t),
 \]
such that
 \[
\lim_{t \to 0^+} u_j(x,t) = f_j(x).
 \]
 We consider the $n$-th convolution
 \be\label{nth}
w(x,t)  =  u_1^{\alpha_1}*u_2^{\alpha_2}*\cdots * u_n^{\alpha_n}(x,t).
 \ee
Then, a natural question arises. Can we fix the diffusion
coefficients in the heat equation in such a way that $w(x,t)$
behaves monotonically in time?  Note that the choice of condition
\fer{sum} is forced by the fact that we want that the monotonicity
of $w(x,t)$, $t >0$ has to hold at least if $u_j(x,t)$ is the
fundamental solution to the heat equation, $j =1,2, \dots, n$. In
this case, in fact, computations are explicit, and, provided
condition \fer{sum} is satisfied,  $w(x,t)$ is increasing in time
independently of the choice of the diffusion coefficients (cf.
Section \ref{he}). In the general case, however, the monotonicity in
time of the $n$-th convolution can be proven under more restrictive
assumptions both on the numbers $\alpha_j$, and only for a unique
choice of the diffusion coefficients $\kappa_j$ (cf. Lemma
\ref{le-young}).

The interest in the monotonicity of the convolution of powers of
solutions to the heat equation is linked to its consequences.
Indeed, the discovery of the monotonicity of $w(x,t)$ for a special
choice of the diffusion coefficients translates immediately to the
proof of  an inequality for convolutions in sharp form. Let $n$ be
an integer, and let $p_j$, $j= 1,\dots, n$, be real numbers such
that $1 \le p_j \le +\infty$ and $\sum_{j=1}^n p_j^{-1} = n-1$. Let
$f_j(x)$, $j= 1,\dots, n$,  be functions on $\R^d$, $ d\ge 1$, such
that $f_j \in L^{p_j}(\R^d)$. In Theorem \ref{th-young} we will show
that the monotonicity of $w(x,t)$ implies the following inequality
for convolutions:
 \be\label{you}
\sup_x | f_1*f_2*\cdots * f_n| \le \prod_{j=1}^n C_{p_j}^d \|
f_j\|_{p_j}.
 \ee
In \fer{you}, the constant $C_p$ which defines the sharp constant is
given by
 \be\label{c+}
 C_p^2 = \frac{p^{1/p}}{p'^{1/p'}},
  \ee
where primes always denote dual exponents, $1/p + 1/p' = 1$. Also,
the expression of the best constant in \fer{you}, in the case in
which the functions $f_j$ are probability density functions, is
obtained by assuming that the functions $f_j$ are suitable Gaussian
densities \cite{Li}. This expression naturally appears in this
monotonicity approach by considering that for large times the
solution to the heat equation behaves as the self-similar Gaussian
profile.

Alternatively, \fer{you} is equivalent to
 \be\label{you2}
 \left| \int f_1(x_1)f_2(x_1-x_2)\cdots f_n(x_{n-1})\, dx_1dx_2\cdots dx_{n-1}
 \right| \le \prod_{j=1}^n C_{p_j}^d \|
f_j\|_{p_j},
 \ee
which is a particular case of the general inequalities obtained by
Brascamp and Lieb \cite{BL},  which are nowadays known as the
Brascamp--Lieb inequalities.

Note that inequality \fer{you} is closely related to the
monotonicity property of the functional given by $L^\infty$-norm of
the $n$-th convolution $w(x,t)$. Naturally one could ask if a
similar property holds for the $L^r$-norm of $w(x,t)$, where $r \ge
0$. Also in this case, the monotonicity in time  can be proven under
suitable assumptions both on the numbers $\alpha_j$, and only for a
unique choice of the diffusion coefficients $\kappa_j$. The study of
the monotonicity in time of $\|w(t)\|_r$ is connected with the
classical Young's inequality  in sharp form ($r>1$), or with its
reverse form ($r<1$).

Last, the limiting cases $r \to 1$ and $r \to 0$ lead to the
monotonicity in time of Shannon's entropy   and of the Renyi entropy
of order $0$ \cite{CTh}.  The monotonicity here leads to the entropy
power inequality of Shannon \cite{Sha}, and to the
Pr\'ekopa--Leindler inequality \cite{Lei,Pr1,Pr2}, respectively.

Therefore, all these well-known functional inequalities can be seen
into a unified framework, as consequences of the monotonicity in
time of the $n$-convolution of powers of solutions to the heat
equation.

As noticed in \cite{Tos3}, the heat equation started to be used as a
powerful instrument to obtain mathematical inequalities in sharp
form in the years between the late fifties to mid sixties. To our
knowledge, the first application of this idea can be found in two
papers  by Linnik \cite{Lin} and Stam \cite{Sta} (cf. also Blachman
\cite{Bla}), published in the same year and concerned with two
apparently disconnected arguments. Stam \cite{Sta} was motivated by
the finding of a rigorous proof of Shannon's entropy power
inequality \cite{Sha}, while Linnik \cite{Lin} used the information
measures of Shannon and Fisher in a proof of the central limit
theorem of probability theory. Also, in the same years, the heat equation has been used in the context of kinetic theory of rarefied gases by McKean \cite{McK} to investigate that large-time behaviour of Kac caricature of a Maxwell gas. There, various monotonicity properties of the derivatives of Shannon's entropy along the solution to the heat equation have been enlightened.

The huge potentialities of the use of the heat equation to prove
inequalities have been rediscovered in more recent times by Carlen,
Lieb and Loss \cite{Car}, that first introduced a Lyapunov
functional of solutions to the heat equation which allows to prove
Young's inequality and its converse for functions of one variable.
Later on, Bennett Carbery Christ and Tao \cite{B1} were able to
extend the result in \cite{Car} to general functions. Other very
closely-related works can be found in papers of Bennett and Bez
\cite{BB},  Borell \cite{Bor}, Barthe and Cordero-Erausquin
\cite{BC} and Barthe-Huet \cite{BH}. In particular, Young's
inequality and its converse have been proven by Bennett and Bez
\cite{BB} by showing that a suitable functional of the convolution
of powers to the solution to the heat the heat equation exhibits
monotonicity properties.

As often happens, however, the seminal ideas of Stam \cite{Bla, Sta}
remained confined within the framework of information theory, where, however,  functional inequalities gained a lot of interest, in reason of their connections with properties of Shannon's and Renyi's entropies \cite{DCT}. A notable exception to this confinement
is a recent paper by Gardner \cite{Gar},
that clarifies the relationship between the Brunn-Minkowski
inequality and other inequalities in geometry and analysis. In
\cite{Gar}, clear connections between the entropy power inequality
of information theory and Young's inequality and others are
described in details, together with an exhaustive list of
references.

As far as the classical Young's inequality is concerned, the
original  proof of the sharp form is due to Beckner \cite{Bec} and
Brascamp and Lieb \cite{BL}. In \cite{BL} Brascamp and Lieb also
proved the sharp form of Young inequality also in the so-called
reverse case.
A different  proof of this sharp reverse Young inequality was subsequently
done by Barthe \cite{Ba}. In their recent paper,  Young's inequality
has been seen in a different light by Bennett and Bez \cite{BB} (cf.
also \cite{B2,B1,Car}). In this paper, Young's inequality is
derived by looking at the monotonicity properties of a suitable functional of the
convolution of powers to the solution to the heat the heat equation.
In this respect, the arguments of \cite{BB} are close to the present
ones.

The connections of the sharp form of Young's inequality with the
Pr\'ekopa--Leindler inequality  has been enlightened by Brascamp and
Lieb \cite{BL}.  Then, the connection of Young's inequality with
Shannon's entropy power inequality has been noticed by Lieb
\cite{Lieb}.

%%%%%%%%%%%%%%%%%%%%%%%%%%%%%%%%%%%%%%%%%%%%%%%%%%%%%%%%%%%%%%%%%%%%%%%%%%%%%%

\section{Heat equation, Lyapunov functionals and dilation
invariance}\label{he}

We begin by recalling some properties of the solution to the heat
equation in $\R^d$, $d \ge 1$
 \be\label{heat}
\frac{\partial u(x,t)}{\partial t} =  \kappa \Delta u(x,t),
 \ee
where $\kappa>0 $ is the (constant) diffusion coefficient. In the
rest of the paper, for the sake of simplicity we will assume that
the initial datum is a non-negative integrable function $f(x)$, that
is
 \be\label{den}
\int_{\R^d} f(x) \, dx = \mu < +\infty.
 \ee
This assumption will not affect the generality of the results that
follow.  The solution to equation \fer{heat} is given by the
function $u(x,t)= f*M_{2\kappa t}(x)$, convolution of the initial
datum with the fundamental solution $M_{2\kappa t}$, where $M_\sigma(x)$, for $\sigma >0$,
denotes the Gaussian density in $\R^n$ of mean $0$ and variance
$d\sigma$
 \be\label{max}
M_\sigma(x) = \frac 1{(2\pi
\sigma)^{d/2}}\exp\left(-\frac{|x|^2}{2\sigma}\right).
 \ee
For large times,  the solution to the heat equation approaches the
fundamental solution. This large-time behaviour can be better
specified by saying that the solution to the heat equation
\fer{heat} satisfies a property which can be defined as the
\emph{central limit property}. If
\begin{equation}\label{FPscal}
 U(x,t) = \left(\sqrt{1+2t}\right)^d\, u(x\, \sqrt{1+2t}, t).
\end{equation}
$U(x,t)$ tends towards a limit function as time goes to infinity,
and this limit function is a Gaussian function
 \be\label{limi}
\lim_{t\to \infty} U(x,t) = M_\kappa(x) \, \int_{\R^n} f(x) \, dx =
\mu M_\kappa(x).
 \ee
This property can be achieved easily by resorting to Fourier
transform, or by exploiting the relationship between the heat
equation and the Fokker--Planck equation \cite{CT} (cf. also
\cite{BBDE} for recent results and references). We note that the
passage $u(x,t) \to U(x,t)$ defined by \fer{FPscal} is mass
preserving, that is
 \be\label{mp}
\int_{\R^d} U(x,t) \, dx =   \int_{\R^d} u(x,t) \, dx.
 \ee
An important remark concerns the necessity to introduce condition
\fer{sum} in our analysis. Since the fundamental solution is a
Gaussian probability density, it is closed under the operation of
convolution \cite{LR}, namely
 \[
 M_{\sigma_1}*M_{\sigma_2}(x) = M_{\sigma_1 + \sigma_2}(x).
 \]
Hence, if we consider at  time $t>0$ the convolution of $n$ powers
of the fundamental solutions of heat equations with diffusion
coefficients $\kappa_j$, $j=1,2,\dots, n$, we obtain
 \[
 M_{2\kappa_1t}^{\alpha_1}*M_{2\kappa_2t}^{\alpha_2}*\cdots *M_{2\kappa_nt}^{\alpha_n} =
 \]
 \[
 \prod_{j=1}^n \left( 4\pi \kappa_j t \right)^{-\alpha_j/2}\left( 4\pi \frac{\kappa_j}{\alpha_j} t \right)^{1/2}M_{2t\kappa_1/\alpha_1 }*M_{2t\kappa_2/\alpha_2}*\cdots *M_{2t\kappa_n/\alpha_n} =
 \]
 \[
 \prod_{j=1}^n \left( 4\pi \kappa_j t \right)^{-\alpha_j/2}\left( 4\pi \frac{\kappa_j}{\alpha_j} t \right)^{1/2}M_{2\Sigma t},
 \]
 where
 \[
 \Sigma = \sum_{j=1}^n \frac{\kappa_j}{\alpha_j}.
 \]
In the expression above the time-dependent quantity in front of the exponential is given by
 \[
\phi(t) = t^{ - \frac 12 \sum_{j=1}^n \alpha_j + \frac 12 (n -1)}.
 \]
Therefore, if the exponents $\alpha_j$ satisfy condition \fer{sum},
independently of the values of the diffusion coefficients
$\kappa_j$, $\phi(t) = 1$, and
 \be\label{co2}
 M_{2\kappa_1t}^{\alpha_1}*M_{2\kappa_2t}^{\alpha_2}*\cdots *M_{2\kappa_nt}^{\alpha_n} =
 \Sigma_1 \exp\left\{- |x|^2/ 4\Sigma t \right\},
  \ee
 where $\Sigma_1$ denotes the constant
  \[
 \Sigma_1=  \left( \frac{\kappa_j}{\alpha_j} \right)^{n/2}\Sigma^{-1/2} \prod_{j=1}^n \left(\kappa_j \right)^{-\alpha_j/2}.
 \]
 Consequently, independently of the values of the diffusion coefficients $\kappa_j$, if the exponents $\alpha_j$ satisfy condition \fer{sum}, for every $x \in \R^d$
 \be\label{d+}
 \frac d{dt} M_{2\kappa_1t}^{\alpha_1}*M_{2\kappa_2t}^{\alpha_2}*\cdots *M_{2\kappa_nt}^{\alpha_n} \ge 0.
 \ee
This property is obviously restricted to a set of positive constants $\alpha_j$ satisfying \fer{sum}.

A second argument is to use the evolution equation for a power of
the solution to the heat equation. If $\alpha >0$ is a positive
constant, and $u(x,t)$ solves \fer{heat}, then $u^\alpha(x,t)$
solves
 \be\label{heat1} \frac{\partial
u^\alpha(x,t)}{\partial t} =  \kappa \left[ \Delta u^\alpha(x,t) +
\alpha(1-\alpha)u^\alpha(x,t)|\nabla \log u(x,t)|^2 \right].
 \ee
Equation \fer{heat1} is particularly adapted to work with
convolutions of powers. Note that equation \fer{heat1} connects in a
natural way dual exponents. In fact, if $\alpha = 1/p$, with $p
>1$, equation \fer{heat1} takes the form
 \[ \frac{\partial
u^{1/p}(x,t)}{\partial t} =  \kappa \left[ \Delta u^{1/p}(x,t) +
\frac 1{pp'}u^{1/p}(x,t)|\nabla \log u(x,t)|^2 \right].
 \]
Our last ingredient is to consider the evolution in time of Lyapunov
functionals of solutions to the heat equation which are
\emph{dilation invariant}, that is invariant with respect to the
scaling
 \be\label{scal}
 f(x) \to f_a(x) = {a^d} f\left( {a} x \right), \quad a >0,
 \ee
In reason of \fer{limi}, this property allows to reckon immediately
the (bounded) limit value of the underlying functional, as time goes
to infinity.

One simple example will clarify why dilation invariance is a key ingredient of our strategy.
Given a solution to the heat equation \fer{heat} let us consider its (finite) Shannon's entropy
 \[
 H(u(t)) = -\int_{\R^d} u(x,t) \log u(x,t) \, dx
 \]
A simple computation shows that the time derivative of $H(u(t))$ is
non-negative \cite{CTh}, and it converges to infinity as time goes
to infinity. Indeed, this happens because Shannon's entropy is not
scaling invariant
 \be\label{h-scale}
H(u_a) = H(u) - d \log a.
 \ee
Clearly, there are various ways to obtain the
scaling invariance of $H$ by adding or multiplying it by suitable
quantities. We resort here to the second moment of $u$. It is easily
checked that the second moment of a probability density function scales
according to
 \be\label{sec}
E(u_a) = \int_{\R^d} |x|^2 u_a(x) \, dx = \frac 1{a^2} E(u).
 \ee
Hence, if the probability density has bounded second moment, a scaling
invariant functional is obtained by coupling Shannon's entropy of $u$ with
the logarithm of the second moment of $u$
 \be\label{ent2}
\Gamma(t) = H(u(t)) - \frac d2 \log E(u(t)).
 \ee
Explicit computations then show that the functional $\Gamma(t)$ is
monotone increasing, but, by virtue of the \emph{central limit
property}, it will converge to a bounded value \cite{Tos3}
 \[
 \Gamma(u(t)) \le \Gamma({M_1}) =  \frac d2
\log \frac{2\pi e}d.
 \]
The rest of the paper will be devoted to the proof of various
inequalities for convolutions in sharp form. For the sake of
simplicity, we will present most of the proofs in dimension $d=1$. The
corresponding higher-dimensional inequalities can be deduced as well
by making use of standard properties of the Gaussian function.

%%%%%%%%%%%%%%%%%%%%%%%%%%%%%%%%%%%%%%%%%%%%%%%%%%%%%%%%%%%%%%%%%%%%%%%%%%%%%%%%%%%%%%
\section{The monotonicity of convolutions}\label{se-you}

Let $n$ be an integer, and let $p_j$, $j= 1,\dots, n$, be real
numbers such that
 \be\label{coef+}
1 \le p_j \le +\infty; \quad \sum_{j=1}^n \frac 1{p_j} = n-1.
 \ee
 Let
$f_j(x)$, $j= 1,\dots, n$,  be non-negative functions on $\R^d$, $
d\ge 1$, such that $f_j \in L^{p_j}(\R^d)$. For any given $j$, $j=
1,\dots, n$, we denote by $u_j(x,t)$ the solution to the heat
equation \fer{heat} with the diffusion coefficients $\kappa_j$
 \be\label{hu}
\frac{\partial u_j(x,t)}{\partial t} =  \kappa_j \Delta u_j(x,t),
 \ee
such that
 \be\label{id}
\lim_{t \to 0^+} u_j(x,t) = f_j(x).
 \ee
The following Lemma shows that there is  a (unique)  choice of the
diffusion coefficients in the heat equation such that $w(x,t)$
behaves monotonically in time.

\begin{lem}\label{le-young}  Let $w(x,t)$ be the $n$-th convolution
 \be\label{ke1}
w(x,t) = u_1^{1/p_1}*u_2^{1/p_2}*\cdots * u_n^{1/p_n}(x,t)
 \ee
where the functions $u_j(x,t)$, $j=1,2,\dots,n$, are solutions to
the heat equation corresponding to the initial values $0 \le
f_j(x)\in L^1(\R^d)$. Then, if for each $j$ the exponents $p_j$
satisfy conditions \fer{coef+} and the diffusion coefficients are
given by $\kappa_j = (p_jp'_j)^{-1}$, $w(x,t)$ is monotonically
increasing in time from
 \[
w(x, t=0) = f_1^{1/p_1}*f_2^{1/p_2}*\cdots * f_n^{1/p_n}(x).
 \]
Moreover, $w(x,t)$ remains constant in time if and only if $f_j(x)$,
$j=1,2,\dots, n$, is a multiple of a Gaussian density of variance
$d\kappa_j$.
\end{lem}

\begin{proof}
For the sake of simplicity, we will prove the Lemma for $d=1$. As
the proof shows, however, analogous computations can be done in
higher dimension.

Since $\sum_{j=1}^n p_j^{-1} = n-1$, H\"older inequality implies
that
 \[
\left| \int f_1(x_1)^{1/p_1}\cdots f_n(x_{n-1})^{1/p_n}\,
dx_1dx_2\cdots dx_{n-1}
 \right| \le \prod_{j=1}^n  \left( \int_{\R} | f_j(x)|\, dx
 \right)^{1/p_j}.
 \]
Hence
 \be\label{ho}
f_1^{1/p_1}*f_2^{1/p_2}*\cdots * f_n^{1/p_n}(x) \le \prod_{j=1}^n
\left( \int_{\R} | f_j(x)|\, dx \right)^{1/p_j},
 \ee
and, since the right-hand side of \fer{ho} depends only on the
$L^1$-norms of the functions, which are preserved by the heat
equation, the function $w(x,t)$ is bounded for all subsequent times
$t>0$. Also, using basic considerations on the heat equation, it is
sufficient to prove the increasing property of $w(t)$ for very
smooth initial data $f_j$, $j=1,2,\dots, n$, with fast decay at
infinity. In order not to worry about derivatives of logarithms,
which will often appear in the proof, we may also impose that
$|\frac d{dx}\log f_j(x)| \le C(1 + |x|^2)$ for some positive
constant $C$. The general case will follow by density \cite{LT}.

For a given $x \in \R$, let us evaluate the time derivative of the
$n$-th convolution $w(x,t)$. We obtain

\be\label{ott}
 \frac{\partial w(x,t)}{\partial t}  = \left(\sum_{j=1}^n
\kappa_j \right) \frac{\partial^2 w(x,t) }{\partial x^2 } +
\sum_{j=1}^n \frac{\kappa_j}{p_jp'_j} R_j(x,t),
 \ee
where, for $j = 1,2, \dots, n$
 \be\label{rj}
R_j(x) = \int u_1(x-x_1)^{1/p_1}\cdots u_n(x_{n-1})^{1/p_n} \left|
\frac{\partial \log u_j}{\partial x}(x_{j-1}- x_j) \right|^2\,
dx_1\cdots dx_{n-1}
 \ee
Indeed,
 \[
\frac{\partial w}{\partial t} = \frac{\partial u_1^{1/p_1}}{\partial
t}*u_2^{1/p_2}*\cdots * u_n^{1/p_n} + u_1^{1/p_1}*\frac{\partial
u_2^{1/p_2}}{\partial t}*\cdots * u_n^{1/p_n} + \dots
\]
\[
 +  u_1^{1/p_1}* u_2^{1/p_2}*\cdots * \frac{\partial
u_n^{1/p_n}}{\partial t},
 \]
and the time derivative of each term on the right-hand side can be
evaluated by considering that the functions $u_j(x,t)$,
$j=1,2,\dots,n$ satisfy the heat equation \fer{hu} (with diffusion
coefficients $\kappa_j$, $j=1,2,\dots,n$). Hence
 \[
\frac{\partial u_1^{1/p_1}}{\partial t}*u_2^{1/p_2}*\cdots *
u_n^{1/p_n} = \kappa_1 \frac{\partial^2 u_1^{1/p_1} }{\partial x^2
}*u_2^{1/p_2}*\cdots * u_n^{1/p_n} +
\]
\[
\frac {\kappa_1}{p_1p'_1}\left(\left| \frac{\partial\log
u_1}{\partial x}\right|^2u_1^{1/p_1}\right)*u_2^{1/p_2}*\cdots *
u_n^{1/p_n}=
\]
 \be\label{112}
\kappa_1 \frac{\partial^2 r }{\partial x^2 }+ \frac
{\kappa_1}{p_1p'_1}\left(\left| \frac{\partial\log u_1}{\partial
x}\right|^2u_1^{1/p_1}\right)*u_2^{1/p_2}*\cdots * u_n^{1/p_n}.
 \ee
An analogous formula holds for the other indexes $j \ge 2$. Note
that in \fer{112}  we used the convolution property
 \be\label{con1}
\frac{\partial^2}{\partial x^2} f*g(x) = \int f''(x-y)g(y)\, dy =
\int f'(x-y)g'(y)\, dy = \int f(x-y)g''(y)\, dy.
 \ee
By property \fer{con1}, it holds that, for each pair of indexes
$(i,j)$ with  $i,j = 1,2, \dots, n$
 \[
\left(f_1*f_2*\cdots*f_n \right)'' =
 \]
 \[
\int f_1(x-x_1)\dots f_n(x_{n-1})(\log f(x_{i-1}-x_i))'(\log
f(x_{j-1}-x_j))'\, dx_1\dots dx_{n-1}.
 \]
Hence, if we take a set of positive constants  $a_{i,j}$'s, $i,j =
1,2, \dots, n$, such that $\sum_{i\not=j}a_{i,j}= 1$, we can express
the second derivative of a convolution as
 \[
\left(f_1*f_2*\cdots*f_n \right)'' = \sum_{i\not=j}a_{i,j}\int
f_1(x-x_1)\dots f_n(x_{n-1})\cdot
 \]
 \[
 \cdot(\log f(x_{i-1}-x_i))'(\log f(x_{j-1}-x_j))'\, dx_1\dots dx_{n-1}.
 \]
This shows that, for any set of positive values $a_{i,j}$ such that
$\sum_{i\not=j}a_{i,j}= 1$, one has
 \[
\frac{\partial^2 w }{\partial x^2 } =
\sum_{i\not=j}\frac{a_{i,j}}{p_ip_j}\int u_1^{1/p_1}(x-x_1)\dots
u_n^{1/p_n}(x_{n-1})\cdot
 \]
 \be\label{fin1}
\cdot(\log u(x_{i-1}-x_i))'(\log u(x_{j-1}-x_j))'\, dx_1\dots
dx_{n-1}.
 \ee
Finally, by setting, for $j=1,2,\dots,n$
 \be\label{log1}
L_j = \log u_j (x_{j-1}-x_j))',
 \ee
we can rewrite \fer{ott} in the following way:
 \[
 \frac{\partial w(x,t)}{\partial t}  = \int  u_1^{1/p_1}(x-x_1)\dots
u_n^{1/p_n}(x_{n-1})\cdot
 \]
 \be\label{ris}
\left( \sum_{j=1}^n \frac{\kappa_j}{p_jp'_j}L_j^2 +  \sum_{l=1}^n
\kappa_l\sum_{i\not=j}\frac{a_{i,j}}{p_ip_j}L_iL_j\right) dx_1\dots
dx_{n-1}.
 \ee
The sign of the time derivative of $w(x,t)$ depends on the quantity
 \be\label{l1}
\mathfrak{L}(u_1, \cdots u_n) = \sum_{j=1}^n
\frac{\kappa_j}{p_jp'_j}L_j^2 + \sum_{l=1}^n
\kappa_l\sum_{i\not=j}\frac{a_{i,j}}{p_ip_j}L_iL_j.
 \ee
Let us set the coefficient of diffusion $\kappa_j = (p_jp'_j)^{-1}
$, and define $Q_j = L_j/p_j$, for $j= 1,2,\dots n$. Then
 \be\label{l2}
\mathfrak{L} = \sum_{j=1}^n \left(\frac{1}{p'_j}\right)^2 Q_j^2 +
\sum_{l=1}^n \frac 1{p_lp'_l}\sum_{i\not=j}a_{i,j}Q_iQ_j.
 \ee
Now, recall that
 \[
\sum_{j=1}^n \frac{1}{p_j} = n-1
 \]
implies that, for all  $j= 1,2,\dots n$
 \[
\frac 1{p_j} = \sum_{i\not=j} \frac{1}{p'_i}.
 \]
 Consequently
 \[
\sum_{l=1}^n \frac 1{p_lp'_l} = \sum_{i\not=j}\frac1{p'_ip'_j}.
 \]
Therefore
  \be\label{l3}
\mathfrak{L} = \sum_{j=1}^n \left(\frac{1}{p'_j}\right)^2 Q_j^2 +
\sum_{i\not=j}\frac1{p'_ip'_j} \sum_{i\not=j}a_{i,j}Q_iQ_j.
 \ee
If we now choose, for $i\not=j$
 \be\label{aij}
a_{i,j} = \frac{(p'_ip'_j)^{-1}}{\sum_{i\not=j}(p'_ip'_j)^{-1}},
 \ee
which is such that $\sum_{i\not=j}a_{i,j} = 1$, we end up with
 \be\label{e3}
\mathfrak{L} = \sum_{j=1}^n \left(\frac{1}{p'_j}\right)^2 Q_j^2 +
\sum_{i\not=j}\frac1{p'_ip'_j}Q_iQ_j = \left( \sum_{j=1}^n
\frac{Q_j}{p'_j} \right)^2 \ge 0.
 \ee
The previous argument shows that the time derivative of $w(x,t)$ can
be made non-negative  by suitably choosing the diffusion
coefficients $\kappa_j$,  $j= 1,2,\dots n$.

Recalling the definition of $Q_j$ (respectively $L_j$), equality to
zero in \fer{e3} holds if and only if
 \be\label{q5}
\frac{1}{p_1p'_1}(\log{u_1(x-x_1)})' + \sum_{j=2}^{n-1}
\frac{1}{p_jp'_j}(\log{u_j(x_{j-1}-x_j)})' +
\frac{1}{p_np'_n}(\log{u_n(x_{n-1})})' =0.
 \ee
As each variable $x_i$  appears as argument of a pair of functions
only, it holds that, for every $i= 1,2, \dots, n-1$
 \be\label{q6}
 \frac{1}{p_{j}p'_{j}}\frac {\partial}{\partial x_{j}} (\log{u_{j}(x_{j-1}-x_{j})})' +
 \frac{1}{p_{j+1}p'_{j+1}} \frac {\partial}{\partial x_{j}} (\log{u_{j+1}(x_{j}-x_{j+1})})' = 0.
 \ee
In \fer{q6} we set $x_0 = x$ and $x_{n} = 0$. On the other hand,
since
 \[
 (\log{u_{j}(x_{j-1}-x_{j})})' = \frac {\partial}{\partial x_{j-1}} \log{u_{j}(x_{j-1}-x_{j})} = -\frac {\partial}{\partial x_{j}} \log{u_{j}(x_{j-1}-x_{j})},
 \]
equation \fer{q6} coincides with
 \be\label{q7}
 \frac{1}{p_{j}p'_{j}}\frac {\partial^2}{\partial x_{j-1}^2} \log{u_{j}(x_{j-1}-x_{j})}=
 \frac{1}{p_{j+1}p'_{j+1}} \frac {\partial^2}{\partial x_{j}^2} \log{u_{j+1}(x_{j}-x_{j+1})}.
 \ee
Note that \fer{q7} can be verified if and only if the functions on
both sides are constant. Thus, there is equality in \fer{q7} if and
only if
 \be\label{m2}
\log u_{j}(x) = c\kappa_j x^2 + c_1 x + d_1, \quad \log u_{j+1}(x) =
c\kappa_j x^2 + c_2 x + d_2 .
 \ee
In other words, there is equality in \fer{q7} if and only if $u_j$
and $u_{j+1}$ are multiple of Gaussian densities, of variances
$c({p_{j}p'_{j}})^{-1}$ and $c({p_{j+1}p'_{j+1}})^{-1}$,
respectively, for any given positive constant $c$. Therefore,
equality in \fer{e3} holds if and only if each function $u_j(x)$, $j
= 1,2,\dots, n$ is a multiple of a Gaussian density of variance
$c({p_{j}p'_{j}})^{-1}$.

\noindent Finally, with this choice of the diffusion coefficients,
for every $x \in \R$ and $t_1< t_2$,
 \be\label{q1}
u_1^{1/p_1}*u_2^{1/p_2}*\cdots * u_n^{1/p_n}(x,t_1) <
u_1^{1/p_1}*u_2^{1/p_2}*\cdots * u_n^{1/p_n}(x,t_2),
 \ee
unless all initial data are multiple of Gaussian densities with the
right variances. Clearly, \fer{q1} is equivalent to say that the
$n$-th convolution $w(x,t)$ is monotone increasing.
As identical proof holds in higher dimension.
This concludes
the proof of the Lemma.
\end{proof}

\begin{rem}
The result of Lemma \ref{le-young} remains true if each diffusion
coefficient $k_j$ is multiplied by a positive constant $C$. In this
case, equality holds if the functions $f_j$ are Gaussian functions
with variances $Cdk_j$.
\end{rem}

\begin{rem}
As already specified in the introduction, our quantity $w(x,t)$ is
related to a particular geometric Brascamp--Lieb inequality. Results
concerning more general Brascamp--Lieb inequalities that are related
to Lemma \ref{le-young} have been obtained by Bennett, carbery,
Christ and Tao in \cite{B1}. This clearly indicates that the proof
of Lemma \ref{le-young} presented here could be extended to cover
more general situations.
\end{rem}

\noindent Lemma \ref{le-young} has important consequences. Indeed,
let us introduce the functional
 \be\label{key1}
\Psi(t) = \sup_x w(x,t) = \sup_x u_1^{1/p_1}*u_2^{1/p_2}*\cdots *
u_n^{1/p_n}(x,t).
 \ee
It is a simple exercise to verify that, in view of conditions
\fer{coef+} on the constants $p_j$, the functional $\Psi(t)$ is
dilation invariant. In reason of this property we prove:

\begin{thm}\label{th-young}  Let $\Psi(t)$ be the  functional
\fer{key1}, where the functions $u_j(x,t)$, $j=1,2,\dots,n$, are
solutions to the heat equation corresponding to the initial values
$0 \le f_j(x)\in L^1(\R^d)$, $d \ge 1$. Then, if for each $j$ the
exponents $p_j$ satisfy conditions \fer{coef+} and the diffusion
coefficients are given by $\kappa_j = (p_jp'_j)^{-1}$, or by a
multiple of them, $\Psi(t)$ is increasing in time from
 \[
\Psi(0) = \sup_x f_1^{1/p_1}*f_2^{1/p_2}*\cdots * f_n^{1/p_n}(x)
 \]
to the limit value
 \be\label{lim11}
\lim_{t \to \infty} \Psi(t) = \prod_{j=1}^n C_{p_j}^d \left(
\int_{\R^d} | f_j(x)|\, dx \right)^{1/p_j}.
 \ee
The constants $C_{p_j}$  in \fer{lim11} are defined as in \fer{c+}.

Moreover, $\Psi(0) = \lim_{t\to\infty} \Psi(t)$ if and only if
$f_j(x)$, $j=1,2,\dots, n$, is a multiple of a Gaussian density of
variance $cd\kappa_j$,  with $c>0$.
\end{thm}

\begin{proof}
Thanks to Lemma \ref{le-young} we know that the functional $\Psi(t)$
is monotonically increasing from $\Psi(t=0)$, unless the initial
densities are Gaussian functions with the right variances. To
conclude the proof, it remains to show that the functional $\Psi(t)$
converges towards the limit value \fer{lim11} as time converges to
infinity. The computation of the limit value uses in a substantial
way the sca\-ling invariance of $\Psi$. In fact, thanks to the
dilation invariance, at each time $t>0$, the value of $\Psi(t)$ does
not change if we scale each function $u_j(x)$, $j = 1,2,\dots, n$,
according to
\begin{equation}\label{FP}
 u_j(x,t) \to U_j(x,t) = \left(\sqrt{1+2 t}\right)^d\, f(x\, \sqrt{1+2t}, t).
\end{equation}
On the other hand, the \emph{central limit property} \fer{limi}
implies that
 \be\label{limi2}
\lim_{t\to \infty} U_j(x,t) = M_{\kappa_j}(x) \, \int_{\R^d} f_j(x)
\, dx
 \ee
Therefore, passing to the limit one obtains
 \be\label{b3}
\lim_{t \to \infty} \Psi(t) = \prod_{j=1}^n \left( \int_{\R^d} |
f_j(x)|\, dx \right)^{1/p_j}\sup_x
M_{\kappa_1}^{1/p_1}*M_{\kappa_2}^{1/p_2}*\cdots *
M_{\kappa_n}^{1/p_n}(x).
 \ee
Owing to the identity
 \be\label{h1}
 M_{\kappa_j}^{1/p_j}(x) = C_{p_j}^d(2\pi)^{(2p'_j/d)^{-1}}M_{1/{p'_j}},
 \ee
and recalling that $\sum_{j=1}^n (p'_j)^{-1} = 1$, we obtain
 \[
M_{\kappa_1}^{1/p_1}*M_{\kappa_2}^{1/p_2}*\cdots *
M_{\kappa_n}^{1/p_n}(x) =
 \]
 \[
(2\pi)^{-d/2}\prod_{j=1}^n C_{p_j}^d M_1(x) = \prod_{j=1}^n
C_{p_j}^d \exp\{-|x|^2/2\}.
 \]
This implies \fer{lim11}, and concludes the proof of the theorem.
\end{proof}

\begin{rem}
Theorem \ref{th-young} is related to the monotonicity in time of a
dilation invariant functional whose components are solutions to the
heat equation. Therefore, the main importance of the theorem is to
highlight the existence of a new Lyapunov functional related to the
heat equation. This result, however, can be rephrased to give a new
proof of known inequalities in sharp form. Let us set, in Theorem
\ref{th-young}
 \[
 g_j(x) = f_j(x)^{1/p_j},
 \]
for $j=1,2, \dots,n$. Then, it holds
 \be\label{bl1}
 \sup_x g_1*g_2* \cdots * g_n(x) \le \prod_{j=1}^n C_{p_j}^d \prod_{j=1}^n
 \|g_j\|_{{p_j}}.
 \ee
Moreover, since
 \[
\sup_x g_1*g_2* \cdots * g_n(x) \ge \int g_1(-x_1)g_2(x_1 -
x_2)\dots g_n(x_{n-1})\, dx_1\dots dx_{n-1},
 \]
inequality \fer{bl1} implies, under the same conditions on the
constants $p_j$,
 \be\label{bl2}
 \int g_1(x_1)g_2(x_1 -
x_2)\dots g_n(x_{n-1})\, dx_1\dots dx_{n-1} \le \prod_{j=1}^n
C_{p_j}^d \prod_{j=1}^n
 \|g_j\|_{{p_j}}.
 \ee
Inequality \fer{bl2} is a particular case of the inequalities
obtained by Brascamp and Lieb \cite{BL} by a different method.
\end{rem}

\begin{rem}
Clearly, the proof of Theorem \ref{th-young} still holds when $n=2$.
In this case, however, the diffusion coefficients $\kappa_j$, $j
=1,2$ coincide. In fact, when $n =2$, the condition \fer{coef+}
reduces to
 \[
1 \le p_j \le +\infty; \quad \frac 1{p_1} +\frac 1{p_2} =1,
 \]
so that $p_1$ and $p_2$ are dual exponents. Consequently $p'_1 =
p_2$ and $p'_2= p_1$, which imply $\kappa_1 = \kappa_2= \kappa=
(p_1p_2)^{-1}$. But in this case the definition \fer{c+} of the
constant $C_p$ implies $C_{p_1} = 1/C_{p_2}$, and the limit
\fer{lim11} takes the value
 \be\label{lim12}
\lim_{t \to \infty} \Psi(t) = \left( \int_{\R^d} | f_1(x)|\, dx
\right)^{1/p_1}\left( \int_{\R^d} | f_2(x)|\, dx \right)^{1/p_2}.
 \ee
Note that in this case inequality \fer{bl2} reduces simply to the
classical H\"older inequality.
\end{rem}

\begin{rem}
As  noticed by Brascamp and Lieb \cite{BL}, Theorem \ref{th-young}
contains as special case the best possible improvement to Young's
inequality. If $n = 3$ \fer{bl2} reads
 \be\label{y2}
\int_{\R^{2d}} f(x) g(x-y) h(y) dx\, \dy \le
(C_pC_qC_s)^d\|f\|_{L^p}\|g\|_{L^q}\|h\|_{L^s},
 \ee
where $1 \le p,q,s \le \infty$, $1/p + 1/q + 1/s = 2$, and equality
holds when $f,g,h$ are suitable Gaussian functions. Choosing
 \[
h(y) = \left( f*g(y)\right)^{r-1}
 \]
leads to an equivalent form of \fer{y2}
 \be\label{y3}
 \| f*g\|_{L^r} \le (C_pC_qC_{r'})^d\|f\|_{L^p}\|g\|_{L^q},
 \ee
namely the standard form of Young's inequality \cite{Bec, BL}.

Also, repeated applications of \fer{y3} give
 \be\label{y4}
\|g_1*g_2*\cdots *g_n \|_r \le C_{r'}^d \prod_{j=1}^n C_{p_j}^d
\|g_j\|_{{p_j}},
 \ee
where $1 \le p_j\le \infty$ and $\sum_{j=1}^n 1/p_j = n-1 + 1/r$.
\end{rem}

\section{Further Lyapunov functionals}

Theorem \ref{th-young} shows the monotonicity in time of the
$L^\infty$-norm of the $n$-th convolution of type \fer{key1}, as
well as its convergence towards an explicitly computable limit value
(in terms of the initial data). The key point in getting this result
was the dilation property of the functional $\Psi(t)$.

To get a similar result for the $L^r$-norm of the $n$-th convolution
$w(x,t)$, $r>0$, and to obtain the (eventual) limit value, we need
that the dilation property still holds for $\| w(t)\|_r$. By
applying the same scaling $u_j(x) \to V_j(x) = a^d V(ax)$ to each
function $u_j(x)$ in \fer{nth} we get
 \[
 V_1*V_2* \cdots *V_n(x) = a^{d\gamma}u_1*u_2* \cdots * u_n(ax) = a^{d\gamma} w(ax),
 \]
where
 \[
 \gamma = \sum_{j=1}^n \alpha_j - n +1
 \]
Hence
 \[
 \int_{\R^d} \left(  V_1*V_2* \cdots *V_n(x)\right)^r \, dx = \int_{\R^d} a^{dr\gamma}w^r(ax)\, dx,
 \]
and dilation invariance occurs if and only if $r\gamma = 1$, that is
 \be\label{new1}
 \alpha_1 + \alpha_2 + \dots + \alpha_n = n-1 + \frac 1r.
 \ee
By analogy with condition \fer{coef+}, we will satisfy condition \fer{new1} in two separate cases.
The first refers to fix, for
$j= 1,\dots, n$ and $s$,  positive real numbers $p_j$ and $r$ such that
 \be\label{coef-}
 p_j < 1, r < 1; \quad \sum_{j=1}^n \frac 1{p_j} = n-1 + \frac 1r.
 \ee
The second refers to fix, for
$j= 1,\dots, n$ and $s$,  positive real numbers $p_j$ and $r$ such that
 \be\label{co+}
 1 < p_j \le \infty, 1< r \le \infty; \quad \sum_{j=1}^n \frac 1{p_j} = n-1 + \frac 1r.
 \ee
In the following, we will analyze the time behaviour of $\|
w(t)\|_r$ in the case \fer{coef-}. Then, the result for the case
\fer{co+} will follow by the same line of proof.

Condition \fer{coef-} implies that $p'_j < 0$  for all $j = 1,2,\dots, n$, and
 \[
\frac 1{p_j} = \sum_{i\not= j} \frac 1{p'_j}+ \frac 1r.
 \]
Making use of the proof of Lemma \ref{le-young}, let us set, for
$j = 1,2,\dots, n$, the (positive) coefficients of diffusion
 \be\label{n4}
\kappa_j = \frac 1{p_j |p'_j|}.
 \ee
Then, by means of elementary computations we obtain
 \be\label{n5}
\sum_{l=1}^n \kappa_l = \sum_{l=1}^n \frac 1{p_l |p'_l|} =
\sum_{i\not=j}  \frac 1{p_i |p'_j|} + \frac 1{r|r'|}.
 \ee
Since the real numbers $p_j$ now satisfy condition \fer{coef-}, the
quantity \fer{l1} considered in Lemma \ref{le-young}, with the same
choice \fer{aij} of the coefficients $a_{i,j}$  takes the form
 \be\label{e7}
\mathfrak{L} = -\left( \sum_{j=1}^n \frac{Q_j}{p'_j} \right)^2 +
\frac 1{r|r'|}\sum_{i\not=j}a_{i,j}Q_iQ_j.
 \ee
It is evident that in this case we cannot expect that $\mathfrak L$
has a definite sign. However, using expression \fer{e7} into
\fer{ris} we obtain
 \[
\frac{\partial w(x,t)}{\partial t}  = \int  u_1^{1/p_1}(x-x_1)\dots
u_n^{1/p_n}(x_{n-1}) \mathfrak L(u_1, \cdots u_n) \, dx_1\cdots
dx_{n-1} =
 \]
 \be\label{q8}
- \int  u_1^{1/p_1}(x-x_1)\dots u_n^{1/p_n}(x_{n-1})\left(
\sum_{j=1}^n \frac{Q_j}{p'_j} \right)^2\, dx_1\cdots dx_{n-1}  +
\frac 1{r|r'|} \frac{\partial^2 w }{\partial x^2 }.
 \ee
In fact, by formula \fer{fin1}
 \[
\int  u_1^{1/p_1}(x-x_1)\dots u_n^{1/p_n}(x_{n-1})
\sum_{i\not=j}a_{i,j}Q_iQ_j\, dx_1\cdots dx_{n-1}  =
 \]
 \[
\int  u_1^{1/p_1}(x-x_1)\dots u_n^{1/p_n}(x_{n-1})
\sum_{i\not=j}\frac{a_{i,j}}{p_ip_j} L_iL_j \, dx_1\cdots dx_{n-1} =
\frac{\partial^2 r }{\partial x^2 }.
 \]
Consequently, thanks to \fer{q8}
 \[
 \frac d{dt} \int w^r(x,t) \, dx = r \int w^{r-1}(x,t) \frac{\partial w(x,t)}{\partial
 t} \, dx = \frac 1{|r'|}\int w^{r-1}\frac{\partial^2 w }{\partial x^2 }\,
 dx +
 \]
 \[
-r \int w^{r-1} \int u_1^{1/p_1}(x-x_1)\dots u_n^{1/p_n}(x_{n-1})
\left( \sum_{j=1}^n \frac{Q_j}{p'_j} \right)^2 \, dx_1\cdots
dx_{n-1} \, dx =
 \]
 \be\label{q9}
\frac{(1-r)^2} r \int w^{r-2}\left( \frac{\partial w }{\partial
x}\right)^2 - r \int u_1^{1/p_1}(x-x_1)\dots u_n^{1/p_n}(x_{n-1})
\left( \sum_{j=1}^n \frac{Q_j}{p'_j} \right)^2.
 \ee
Surprisingly, the expression on \fer{q9} has a sign. This is
consequence of the following Lemma, which generalizes a similar
result that dates back to Blachman \cite{Bla}. In case of
convolution of two functions, analogous result has been obtained
recently in \cite{Tos1}.

\begin{lem}\label{blac}
Let $w(x)$ be the (smooth) $n$-th convolution defined by \fer{key1}.
Then, for any set of positive constants $p_j$ and $r$, and positive
constants $\lambda_j$, $j = 1,2, \dots ,n$ such that $\sum_{j=1}^n
\lambda_j = 1$ it holds
 \be\label{w1}
\int w^{r-2}\left( \frac{\partial w }{\partial x}\right)^2 \le \int
w^{r-1}(x) \int u_1^{1/p_1}(x-x_1)\dots u_n^{1/p_n}(x_{n-1}) \left(
\sum_{j=1}^n \frac{\lambda_j}{p_j}L_j \right)^2.
 \ee
Moreover, equality in \fer{w1} holds if and only if any function
$u_j$, $j = 1,2, \dots , n$ is  multiple of a Gaussian function of
variance $\lambda_j/p_j$.
\end{lem}

 \begin{proof}
By property \fer{con1},  if we take a set of positive constants
$\lambda_j$, $j = 1,2, \dots ,n$ such that $\sum_{j=1}^n \lambda_j =
1$ we can express the first derivative of $w(x)$ as
 \[
w'(x) =  \int u_1^{1/p_1}(x-x_1)\dots u_n^{1/p_n}(x_{n-1})
\sum_{j=1}^n \frac{\lambda_j}{p_j}L_j dx_1\dots dx_{n-1},
 \]
where $L_j$ is defined as in \fer{log1}. Therefore, by denoting
 \be\label{m3}
d\mu_x(x_1, \dots , x_{n-1}) = \frac{u_1^{1/p_1}(x-x_1)\dots
u_n^{1/p_n}(x_{n-1})}{w(x)},
 \ee
we obtain
 \[
\frac{w'(x)}{w(x)} = \int \sum_{j=1}^n \frac{\lambda_j}{p_j}L_j
d\mu_x(x_1, \dots , x_{n-1}).
 \]
Note that, for any $x \in \R$ the measure $d\mu$ defined in \fer{m3}
is a unit measure on $\R^{n-1}$,
 \[
 \int_{\R^{n-1}} d\mu_x \, dx_1\cdots
dx_{n-1} = 1.
 \]
Jensen's inequality then gives
 \be\label{jen}
\left(\frac{w'(x)}{w(x)}\right)^2 \le \int \left( \sum_{j=1}^n
\frac{\lambda_j}{p_j}L_j \right)^2 d\mu_x(x_1, \dots , x_{n-1}).
 \ee
Multiplying both sides of \fer{jen} by $w^r(x)$, and integrating
over $x$ proves the Lemma.

Note that, since equality in Jensen's inequality holds if and only
if the argument is constant, equality in \fer{jen} holds if and only
if
 \[
\sum_{j=1}^n \frac{\lambda_j}{p_j}L_j = const.
 \]
Hence, the reasoning of the last part of Lemma \ref{le-young} can be
repeated to show that there is equality in \fer{w1} if and only if
any function $u_j$, $j = 1,2, \dots , n$ is  multiple of a Gaussian
function of variance $\lambda_j/p_j$.
 \end{proof}

Let us return to formula \fer{q9}. Conditions \fer{coef-} imply that
 \[
 \sum_{j=1}^n \frac 1{|p'_j|} = \frac 1{|r'|}.
 \]
Hence
 \[
 \frac r{1-r} \sum_{j=1}^n \frac 1{|p'_j|} = 1.
 \]
Choosing then
 \[
\lambda_j = \frac r{1-r}\frac 1{|p'_j|},
 \]
we obtain that \fer{w1} reads
  \[
\int w^{r-2}\left( \frac{\partial w }{\partial x}\right)^2 \le \frac
{(1-r)^2}{r^2}\int w^{r-1}(x) \cdot
 \]
 \be\label{w2}
 \cdot \int u_1^{1/p_1}(x-x_1)\dots u_n^{1/p_n}(x_{n-1})
\left( \sum_{j=1}^n \frac{1}{p'_j}Q_j \right)^2.
 \ee
 This shows that the quantity in \fer{q9} is negative.
Hence, we proved that, if the positive constants $p_j$ and $s$
satisfy conditions \fer{coef-},  the functional
 \be\label{key2}
\Lambda(t) =  \| w(t)\|_{r} = \left( \int
(u_1^{1/p_1}*u_2^{1/p_2}*\cdots
* u_n^{1/p_n})^r(x,t) \, dx \right)^{1/r}
 \ee
is monotone decreasing.  Since we know that, in
view of conditions \fer{coef-} on the constants $p_j$, the
functional $\Lambda(t)$ is dilation invariant, we proved:

\begin{thm}\label{th-you-}  Let $\Lambda(t)$ be the  functional
\fer{key2}, where the functions $u_j(x,t)$, $j=1,2,\dots,n$, are
solutions to the heat equation corresponding to the initial values
$0 \le f_j(x)\in L^1(\R^d)$, $d \ge 1$. Then, if for each $j$ the
exponents $p_j$ satisfy conditions \fer{coef-} and the diffusion
coefficients are given by $\kappa_j = (p_j|p'_j|)^{-1}$,
$\Lambda(t)$ is decreasing in time from
 \[
\Lambda(0) = \left( \int \left(f_1^{1/p_1}*f_2^{1/p_2}*\cdots *
f_n^{1/p_n}(x)\right)^r \, dx \right)^{1/r}
 \]
to the limit value
 \be\label{lim13}
\lim_{t \to \infty} \Lambda(t) = C_{r'}^d\prod_{j=1}^n C_{p_j}^d \left(
\int_{\R^d} | f_j(x)|\, dx \right)^{1/p_j}.
 \ee
The constants $C_{p_j}$  in \fer{lim11} are defined by
 \be\label{c-}
 C_p^2 = \frac{p^{1/p}}{|p'|^{1/p'}},
  \ee
Moreover, $\Lambda(0) = \lim_{t\to\infty} \Lambda(t)$ if and only if
$f_j(x)$, $j=1,2,\dots, n$, is a multiple of a Gaussian density of
variance $d\kappa_j$.
\end{thm}

\begin{proof}
We know that the functional $\Lambda(t)$ is monotonically decreasing
from $\Lambda(t=0)$, unless the initial densities are Gaussian
functions with the right variances. In addition, $\Lambda(t)$ is
dilation invariant. As in Theorem \ref{th-young}, let us scale each
function $u_j(x)$, $j = 1,2,\dots, n$, according to \fer{FP}.
Therefore, by the central limit property, passing to the limit one
obtains
 \be\label{b33}
\lim_{t \to \infty} \Lambda(t) = \prod_{j=1}^n \left( \int_{\R^d} |
f_j(x)|\, dx \right)^{1/p_j} \left\|
M_{\kappa_1}^{1/p_1}*M_{\kappa_2}^{1/p_2}*\cdots *
M_{\kappa_n}^{1/p_n}\right\|_r.
 \ee
The value of the integral can be evaluated by using formula \fer{h1}
of Theorem \ref{th-young}, with the additional difficulty to
evaluate the norm of a Gaussian in $L^r$. Thanks to condition
\fer{coef-} we obtain
 \[
\left\|
M_{\kappa_1}^{1/p_1}*M_{\kappa_2}^{1/p_2}*\cdots *
M_{\kappa_n}^{1/p_n}\right\|_r = C_{r'}^d\prod_{j=1}^n C_{p_j}^d.
 \]
This concludes the proof of the theorem.
\end{proof}

The computations leading to Theorem \ref{th-you-} can be repeated
step-by-step in the case in which the $p_j$'s and $r$ satisfy
condition \fer{co+}. In this case, however, the sign of $\mathfrak
L$ changes, and we obtain

\begin{thm}\label{th-you+}  Let $\Lambda(t)$ be the  functional
\fer{key2}, where the functions $u_j(x,t)$, $j=1,2,\dots,n$, are
solutions to the heat equation corresponding to the initial values
$0 \le f_j(x)\in L^1(\R^d)$, $d \ge 1$. Then, if for each $j$ the
exponents $p_j$ satisfy conditions \fer{co+} and the diffusion
coefficients are given by $\kappa_j = (p_jp'_j)^{-1}$,
$\Lambda(t)$ is increasing in time from
 \[
\Lambda(0) = \left( \int \left(f_1^{1/p_1}*f_2^{1/p_2}*\cdots *
f_n^{1/p_n}(x)\right)^r \, dx \right)^{1/r}
 \]
to the limit value
 \be\label{lim14}
\lim_{t \to \infty} \Lambda(t) = C_{r'}^d\prod_{j=1}^n C_{p_j}^d \left(
\int_{\R^d} | f_j(x)|\, dx \right)^{1/p_j}.
 \ee
The constants $C_{p_j}$  in \fer{lim11} are defined by \fer{c+}.
Moreover, $\Lambda(0) = \lim_{t\to\infty} \Lambda(t)$ if and only if
$f_j(x)$, $j=1,2,\dots, n$, is a multiple of a Gaussian density of
variance $d\kappa_j$.
\end{thm}

\begin{rem} The monotonicity property of the functional $\Lambda(t)$ defined by \fer{key2}
have been noticed first by Bennett and Bez \cite{BB} by means of a
different approach. Consequently, the results of both Theorems
\ref{th-you-} and \ref{th-you+} also follow from their arguments. We
note, however, that the dilation invariance property of
$\Lambda(t)$, which is at the basis of the direct proof of the
Theorems, has not explicitly taken into account before.
\end{rem}

\begin{rem}
Theorems \ref{th-you-} and \ref{th-you+} show the monotonicity
properties of the $L^r$-norm of the $n$-th convolution of powers of
solutions to the heat equation. As discussed at the end of Theorem
\ref{th-young}, apart from its intrinsic physical interest, this
monotonicity can be rephrased in the form of inequalities for
convolutions in sharp form. In particular, when $n=2$, Theorem
\ref{th-you-} contains the sharp form of Young inequality in the
so-called reverse case
  \be\label{young-r}
\| f*g\|_r \ge (C_pC_qC_{r'})^d\| f\|_p\| g\|_q ,
 \ee
where $0 < p,q,r < 1$ while $1/p
+ 1/q = 1 +1/r$, and $C_p$ is defined by \fer{c-}.
\end{rem}

\begin{rem}
A particular case of Theorem \ref{th-you+} implies Babenko's inequality \cite{Bab}
(cf. also Beckner \cite{Bec}):
 \be\label{bab}
  \|\mathfrak F f\|_{q} \le C_{q}^d \|f\|_{q'},
 \ee
where $C_{q}$ is defined as in \fer{c+}, $q$ is an even integer $q =
2,4,6, \dots$, and $\mathfrak F f$ denotes the Fourier transform of
$f$. Here the Fourier transform is defined for integrable functions
by
 \[
\mathfrak F f(\xi)  = \int_{\R^d} \exp \left\{-2\pi i x\cdot\xi\right\}f(x) \, dx
 \]
Inequality \fer{bab} follows by choosing in Theorem \ref{th-you+}
$r= 2$ and $1/p_j = (2n-1)/2n$, which are such that condition
\fer{coef+} is satisfied. In this case, in fact, by setting $f_j =
f$, for $j =1,2,\dots, n$, and $g^q = f$, we obtain that $f$
satisfies the inequality
 \[
 \left( \int ( \ \underbrace{f*f*\cdots *f}_n \ )^2 \,dx \right)^{1/2} \le C_q^{dn} \| f \|_{q}^n.
 \]
 Since
  \[
 \mathfrak F \left( \ \underbrace{f*f*\cdots *f}_n \ \right) = \left(\mathfrak F f\right)^n,
  \]
 by Parceval's identity we conclude that
  \be\label{bab1}
  \left( \int ( \mathfrak F f )^{2n} \,d\xi \right)^{1/2} \le C_q^{dn} \| f \|_{q}^n.
  \ee
We remark that, as explicitly mentioned in \cite{BB}, the
monotonicity of the quantity in \fer{bab1} also follows from the
results in \cite{B1} (cf. also \cite{BB}). A further inside into
Haussdorff--Young inequality, with counterexamples to the
monotonicity of $\|\mathfrak F u^{1/p}(t)\|_{p'}$ whenever $p$ is
not an even integer can be found in \cite{BBC}.
\end{rem}

\section{Monotonicity and Pr\'ekopa--Leindler inequality}

The analysis of the preceding section shows the monotonicity
properties of the $L^r$ norm of the $n$-th convolution of powers of
the solutions to the heat equation. In particular Theorem
\ref{th-young} covers the $L^\infty$ case, while Theorem
\ref{th-you-} (respectively Theorem \ref{th-you+}) cover the case
$r<1$ (respectively $r>1$). Two limit cases remain to be examined,
namely the cases $r\to 0$ and $r\to1$.  Here we will briefly discuss
the first case, leaving the second to the next section.

Given a set of positive constants $q_j$, $j =1,2,\dots, n$, such
that $\sum_{j=1}^n 1/q_j = 1$, and a constant $N \gg 1$, we choose
in Theorem \ref{th-you-}
 \be\label{n1}
 p_j = \frac{q_j}N, \quad r = \frac 1{N-(n-1)}.
 \ee
Then, if $N\ge \max_j{q_j} + n$, conditions \fer{coef-} are
satisfied and the monotonicity of $\Lambda(t)^r$ is guaranteed. By
definition
 \[
 w(x)^r  = \left( u_1^{N/q_1}* u_2^{N/q_2}* \cdots * u_n^{N/q_n}(x,t)\right)^{1/N-n+1} =
  \]
  \[
 \left( \int \left( u_1(x-x_1)^{1/q_1}\cdots u_n(x_{n-1})^{1/q_n}\right)^N \, dx_1 \cdots dx_{n-1} \right)^{1/N-n+1}.
 \]
Hence we obtain
 \be\label{li8}
 \lim_{N \to \infty} \int  w(x)^r \, dx = \int \sup_{x_1,\dots,x_{n-1}} u_1(x-x_1)^{1/q_1}\cdots u_n(x_{n-1})^{1/q_n} \, dx.
 \ee
This implies that, if $\Upsilon_N(t)$ denotes the functional
 \be\label{upsN}
 \Upsilon_N(t) = \left( \int \left( u_1(x-x_1)^{1/q_1}\cdots u_n(x_{n-1})^{1/q_n}\right)^N \, dx_1 \cdots dx_{n-1} \right)^{1/N-n+1},
 \ee
thanks to Theorem \ref{th-you-},  $\Upsilon_N(t)$ is monotonically
decreasing in time, provided the coefficients of diffusion are the
correct ones.

Note that, for any given $N$, the coefficients of diffusions $\kappa_j$ depend on it, and
  \[
 \kappa_j^N = \frac {N(N-q_j)}{q_j^2}.
 \]
On the other hand, Theorem \ref{th-you-} remains true if we multiply
all coefficients of diffusion by the same constant. Therefore,
without affecting the monotonicity of $\Upsilon_N(t)$ we can fix the
coefficients of diffusion as
 \be\label{c4}
\kappa_j^N = \frac {N(N-q_j)}{N^2}\frac 1{q_j^2}.
 \ee
By letting $N \to \infty$ we finally obtain that the functional
 \be\label{ups}
 \Upsilon(t) = \int \sup_{x_1,\dots,x_{n-1}} u_1(x-x_1)^{1/q_1}\cdots u_n(x_{n-1})^{1/q_n} \, dx
 \ee
is monotonically decreasing in time if the coefficients of diffusion
in the heat equations are given by $\kappa_j = 1/q_j^2$. Since the
functional $\Upsilon(t)$ is invariant under dilation, we can pass to
the limit to find the lower bound. By the same argument of the proof
of Theorem \ref{th-young}, we conclude that the limit value is
obtained by setting
 \[
 u_j(x) = \int f_j(x) \, dx M_{1/q_j^2}.
 \]
Explicit computations then show that
 \be\label{fi8}
 \lim_{t \to \infty} \Upsilon(t) = \prod_{j=1}^n q_j^{-d/q_j}\left( \int f_j(x) \, dx. \right)^{1/q_j}
 \ee
By setting $f_j(x) = g_j(q_jx)$, which implies
 \[
 \int_{\R^d} f_j(x)^{1/q_j} \,dx = q_j^{-d/q_j} \int_{\R^d} q_j(x)^{1/q_j} \,dx,
 \]
we conclude with the following

 \begin{thm}\label{th-pl}  Let $\Upsilon(t)$ denote the  functional \fer{ups}
 where the functions $u_j(x,t)$, $j=1,2,\dots,n$, are
solutions to the heat equation corresponding to the initial values
$0 \le f_j(x)\in L^1(\R^d)$, $d \ge 1$. Then, if
exponents $q_j$ satisfy $\sum_{j=1}^n q_j^{-1} = 1$, and the diffusion
coefficients are given by $\kappa_j = q_j^{-2}$,
$\Upsilon(t)$ is decreasing in time from
 \[
\Upsilon(0) = \int \sup_{x_1,\dots,x_{n-1}} f_1(q_1(x-x_1))^{1/q_1}\cdots f_n(q_n(x_{n-1}))^{1/q_n} \, dx
 \]
to the limit value
 \be\label{lim15}
\lim_{t \to \infty} \Upsilon(t) = \prod_{j=1}^n  \left(
\int_{\R^d} | f_j(x)|\, dx \right)^{1/q_j}.
 \ee
Moreover, $\Upsilon(0) = \lim_{t\to\infty} \Upsilon(t)$ if and only if
$f_j(x)$, $j=1,2,\dots, n$, is a multiple of a Gaussian density of
variance $d\kappa_j$.
\end{thm}

\begin{rem}
If $n=2$ the monotonicity of the functional $\Upsilon$  proven in
Theorem \ref{th-pl} implies the classical Pr\'ekopa--Leindler
inequality. In this case, in fact one obtains the
Pr\'ekopa--Leindler theorem \cite{Lei,Pr1,Pr2} that reads
 \[
 \|h\|_1 \ge \|f\|_1^\lambda\|g\|_1^{1-\lambda},
 \]
 where
 \[
 h(x|f,g) = \sup_x f\left( \frac{x-y}\lambda  \right)^{\lambda}g\left( \frac{x-y}{1-\lambda}  \right)^{1-\lambda}.
  \]
The  derivation of Pr\'ekopa--Leindler inequality from the Young's
inequality has been obtained by Brascamp and Lieb \cite{BL}. Our
result, however, enlightens a new meaning of this inequality, that
is viewed as a consequence of the monotonicity of a Lyapunov
functional of the convolution of two powers of the solution to the
heat equation.
 \end{rem}
 \begin{rem}
Theorem \ref{th-pl} is a corollary of the general result of Theorem
\ref{th-you-}. However, a direct proof of monotonicity could be
possible by looking at the functional \fer{ups} directly.
 \end{rem}

\section{A short proof of entropy power inequality}

In its original version, Shannon's entropy power inequality
(\emph{EPI}) \cite{Sha} gives a lower bound on Shannon's entropy
functional of the sum of independent random variables $X, Y$ with
densities
 \be\label{entr}
\exp\left(\frac 2d H(X+Y)\right) \ge \exp\left(\frac 2d H(X)\right)+
\exp \left( \frac 2d H(Y)\right),
 \ee
with equality if $X$ and $Y$ are Gaussian random variables.
Shannon's entropy functional of the probability density function $f(x)$ of $X$ is
 \be
H(X) = H(f) = - \int_{\R^d} f(v) \log f(v)\, dv.
 \ee
Note that Shannon's entropy functional coincides to Boltzmann's
entropy up to a change of sign.  The entropy-power
 \[
 N(X) = N(f) = \exp\left(\frac 2d H(X)\right)
 \]
(variance of a Gaussian random variable with the same Shannon's
entropy functional) is maximum and equal to the variance when the
random variable is Gaussian, and thus, the essence of \fer{entr} is
that the sum of independent random variables tends to be \emph{more
Gaussian} than one or both of the individual components.

The first rigorous proof of inequality \fer{entr} was given by Stam
\cite{Sta} for the case $d=1$ (see also Blachman \cite{Bla} for the
generalization of \emph{EPI} to $d$-dimensional random vectors), and
was based on an identity which couples Fisher's information with
Shannon's entropy functional \cite{CTh}.

Making use of the relationship between mutual information and
minimum mean-square error for additive Gaussian channels \cite{GSV},
a different and simpler proof of \emph{EPI} based on an elementary
estimation--theoretic reasoning which sidesteps invoking Fisher's
information, and makes use of a result of Lieb \cite{Lieb}, was
recently given in \cite{GSV2} (see also Rioul \cite{Rio} for a
unified view of proofs of \emph{EPI} via Fisher's information and
minimum mean-square errors).

Other variations of the entropy--power inequality are present in the
literature. Costa's strengthened entropy--power inequality
\cite{Cos}, in which one of the variables is Gaussian, and a
generalized inequality for linear transforms of a random vector due
to Zamir and Feder \cite{ZF}.

Also, other properties of Shannon's entropy-power $N(f)$ have been
investigated so far. In particular, the \emph{concavity of entropy
power} theorem, which asserts that
 \be\label{conc}
\frac{d^2}{dt^2}\left(N(u(t))\right) \le 0
 \ee
provided that $u(t)$ is the solution  to the heat equation
\fer{heat}. Inequality \fer{conc} is due to Costa \cite{Cos}. Later,
the proof has been simplified in \cite{Dem,DCT}, by an argument
based on the Blachman-Stam inequality \cite{Bla}. More recently, a
short and simple proof has been obtained by Villani \cite{Vil},
using an old idea by McKean \cite{McK}. Various consequences of
inequality \fer{conc}, including the logarithmic Sobolev inequality
and Nash's inequality have been recently discussed in \cite{Tos2}.

As noticed by Lieb \cite{Lieb}, the \emph{EPI} can also be proven as
a limit case of the Young inequality in the sharp form \fer{y3},  by
letting the parameters $p,q$ and $r$ tend to one. This result can be
obtained as follows. Let $0< a < 1$ denote a fixed constant. For a
given (small) positive $\c$, let us consider Young's inequality
\fer{y3} with
 \be\label{pq}
 r=1+\c, \qquad  p = \frac{1+\c}{1+a \c}, \qquad  q = \frac{1+\c}{1+(1-a)
 \c},
 \ee
which are such that
 \[
 \frac 1r +1 = \frac 1p + \frac 1q.
 \]
Note that, as $\c \to 0$, $p,q,r \to 1$. Let  $f,g, h$ be smooth probability
densities, and let us define $z(\c) = \|h\|_{1+\c}$. Then $z(0) = 1$,
and thanks to the identity
 \[
 z(\c) = \exp\left\{ \frac 1{1+\c} \log
\int_{\R^d} h^{1+\c}\, dv \right\},
 \]
one evaluates straightforwardly
 \be\label{der}
 z'(\c) = z(\c)\left[ -\frac 1{(1+\c)^{2}}\log
\int h^{1+\c} + \frac 1{1+\c} \frac{\int h^{1+\c}\log h}{\int
h^{1+\c}}\right].
 \ee
Hence,
 \[
z'(0) = \int_{\R^d} h \log h \, dv = -H(h).
 \]
Owing to the smoothness of $f*g$, we can expand $\|f*g\|_{1+\c} $ in
Taylor's series of $\c$ up to order one, to obtain
 \be\label{left}
\|f*g\|_{1+\c}= 1 - H(f*g)\c + o_1(\c),
 \ee
where $o_1(\c)$ is such that $o_1(\c)/\c \to 0$ as $\c \to 0$.
Analogous computations for the function
 \[
\omega(\c) = \exp\left\{ d\log(C_pC_qC_{r'}) + \frac 1{p} \log
\int_{\R^d} f^{p}\, dv + \frac 1{q} \log \int_{\R^d} g^{q}\,
dv\right\}
 \]
where $p$ and $q$ are defined in \fer{pq}, allow to conclude that
 \be\label{der1}
 \omega'(0) = \frac d2\left( a \log a + (1-a) \log (1-a)\right) -(1-a)H(f) -a
 H(g).
 \ee
Therefore, expanding again in Taylor's series of $\c$, we obtain
 \be\label{right}
\omega(\c) =  1 + \left( \frac d2\left( a \log a + (1-a) \log
(1-a)\right) -(1-a)H(f) -a H(g)\right) \c + o_2(\c),
 \ee
where again $o_2(\c)/\c \to 0$ as $\c \to 0$. It is interesting to
remark that the sharp constant $(C_pC_qC_{r'})^d$ furnishes an
important contribution to formula \fer{left}.  This contribution can
be derived straightforwardly using the identity
 \[
\frac d{d\c}\left(\frac 1p\right) = - \frac d{d\c}\left(\frac 1{p'}\right).
 \]
This gives
  \[
\frac d{d\c} \log C_p^2 = \frac d{d\c} \left( \frac 1p \log p -
\frac 1{p'} \log p' \right) = \left( -2 + \log \frac p{p'}
\right)\frac d{d\c}\left(\frac 1p\right) =
  \]
 \[
\left( -2 + \log (p-1) \right)\frac d{d\c}\left(\frac 1p\right) = \frac{1- a}{(1+\c)^2}\left( 2 - \log \frac{(1-a)\c}{1+a\c} \right),
 \]
and
 \[
\frac d{d\c}\log(C_pC_qC_{r'})^2 = \frac 1{(1+\c)^2} \left( (1-a)
\log \frac{1-a}{1+a\c} + a \log \frac{a}{1+(1-a)\c}\right).
 \]
In conclusion we have the following \cite{Lieb}:

\begin{lem}\label{loc2}
Let the probability densities $f(x)$ and $g(x)$ $x \in \R^d$ possess bounded
Shannon's entropy functional. Then, for any positive constant $0 <
a < 1$ the following inequality holds
 \be\label{EPI1}
 H(f*g) \ge (1-a)H(f) +a H(g) -\frac d2 \left( a \log a + (1-a) \log (1-a)\right).
 \ee
 \end{lem}

\begin{proof}
The proof is a direct consequence of the sharp Young inequality
\fer{y3}. With our notations, Young inequality can be rephrased as
$z(\c) - \omega(\c) \le 0$. Using expansions \fer{left} and
\fer{right}, and letting $\c \to 0$, inequality \fer{EPI1} follows
for smooth densities. A standard density argument then concludes the
proof.
\end{proof}

Shannon's entropy power inequality then follows by maximizing the
right-hand side of inequality \fer{EPI1}. A simple computation shows
that the right-hand side, say $A(a, H(f), H(g))$ attains the maximum
when
 \be\label{max1}
 a = \bar a = \frac{\exp \left\{ 2 \left( H(g) - H(f) \right)/d \right\}}{ 1 + \exp\left\{ 2 \left( H(g) - H(f) \right)/d \right\} },
 \ee
and, for $a=\bar a$
 \be\label{value}
A(\bar a, H(f), H(g)) = \frac d2 \log \left\{ \exp\left(2H(f)/d
\right)+ \exp \left( 2H(g)/d \right) \right\}.
 \ee
With analogous computations, Shannon's entropy-power inequality can
be easily extended to a convolution of $n$ probability densities by
means of Theorem \ref{th-you+}.

While the result of Lieb \cite{Lieb} outlines an interesting
connection between Young's inequality and the entropy power
inequality, the proof of \emph{EPI} via Young's inequality does not
contain any connection with our idea about monotonicity properties
of Lyapunov functionals for the solution to the heat equation.
Indeed, a much simpler direct proof is available by making use of
this idea. For the moment, let us fix the dimension equal to $1$.

Let as usual $w(x,t)$ denote the $n$-th convolution
 \be\label{c3}
 w(x,t) = u_1*u_2* \cdots *u_n(x,t),
 \ee
where the functions $u_j(x,t)$, $j=1,2,\dots,n$, are solutions to
the heat equations, with coefficients of diffusion $\kappa_j$,
corresponding to the initial probability densities $0 \le f_j(x)$
with bounded Shannon's entropy. It is important to note that, in
view of the closure property of the Gaussian density \fer{max} with
respect to convolutions, $w(x,t)$ itself satisfies the heat equation
\fer{heat} with coefficient of diffusion $\kappa = \sum_{j=1}^n
\kappa_j$.  For any set of positive values $\gamma_j$,
$j=1,2,\dots,n$,  such that $\sum_{j=1}^n \gamma_j = 1$, we
introduce the functional
 \be\label{lh}
 \Phi(t) = H(w(t)) - \sum_{j=1}^n \gamma_j H(u_j(t)).
 \ee
Let $f_\alpha$ be the scaled function defined as in \fer{scal}. Since, for $\alpha >0$
 \[
 H(f_\alpha) = H(f) - \log \alpha,
 \]
the functional $\Phi(t)$ is dilation invariant. Given $t>0$, let us evaluate the time derivative of $\Phi(t)$. We obtain
 \be\label{der2}
 \frac d{dt} H(w(t)) = \kappa I( w(t)) - \sum_{j=1}^n \gamma_j \kappa_j I(u_j(t)),
 \ee
where we defined by $I(f)$ the Fisher information of the density $f$, given in any dimension $d \ge 1$ by
 \be\label{fish}
  I(f) = \int_{\R^d} \frac{|\nabla f(x)|^2}{f(x)}\, dx.
 \ee
By setting in Lemma \ref{blac} $r =1$ and $p_j=1$, $j=1,2,\dots,n$, which satisfy conditions \fer{co+},  inequality \fer{w1} assumes the form
 \be\label{f5}
 I(w ) \le \int dx \int u_1(x-x_1)\dots u_n(x_{n-1}) \left(
\sum_{j=1}^n \lambda_jL_j \right)^2=
 \sum_{j=1}^n \lambda_j^2 I(u_j).
 \ee
Formula \fer{f5} follows simply owing to the definition of $L_j$, and applying Fubini's theorem.
The proof of \fer{f5} in the case of the convolution of two functions goes back to Blachman \cite{Bla}.

Thanks to \fer{f5}, by setting the constants $\gamma _j = k_j/k$, we
have at once that these constants satisfy the condition
$\sum_{j=1}^n \gamma_j = 1$, and that the sign of the derivative
\fer{der2}, consequent to this choice, is negative, unless the
functions $u_j$ are Gaussian. Since the functional $\Phi(t)$ is
dilation invariant, we can pass to the limit $t \to \infty$
obtaining
 \be\label{li6}
 \lim_{t\to\infty} \Phi(t) =  H(M_\kappa) - \sum_{j=1}^n \frac{\kappa_j}\kappa H(M_{\kappa_j}).
 \ee
Since
 \[
 H(M_\sigma) = \frac 12 \log 2\pi\sigma,
 \]
 we obtain from \fer{li6}
 \be\label{li7}
  \lim_{t\to\infty} \Phi(t) = - \frac 12 \sum_{j=1}^n \frac{\kappa_j}\kappa \log \frac{\kappa_j}\kappa.
 \ee
Clearly, the same result holds in dimension $d \ge 1$.
Hence we proved the following:

\begin{thm}\label{th-ent}  Let $\Phi(t)$ be the  functional
\fer{key2}, where the functions $u_j(x,t)$, $j=1,2,\dots,n$, are
solutions to the heat equation corresponding to the initial probability densities
$f_j(x)\in L^1(\R^d)$, $d \ge 1$. Then, if the diffusion
coefficients $\kappa_j = C\gamma_j$, $j=1,2,\dots,n$ and $C>0$,
$\Phi(t)$ is decreasing in time from
 \[
\Phi(0) = H(f_1*f_2*\cdots *f_n) - \sum_{j=1}^n \gamma_j H(f_j)
 \]
to the limit value
 \be\label{lim16}
\lim_{t \to \infty} \Phi(t) = - \frac d2 \sum_{j=1}^n \gamma_j \log \gamma_j.
 \ee
Moreover, $\Phi(0) = \lim_{t\to\infty} \Phi(t)$ if and only if
$f_j(x)$, $j=1,2,\dots, n$, is a Gaussian density of
variance $d\kappa_j$.
\end{thm}

Theorem \ref{th-ent} shows the monotonicity of a dilation invariant
functional linked to the Shannon's entropy of a $n$-th convolution
of probability density functions. A direct consequence of this
monotonicity is the entropy power inequality. Indeed, the
monotonicity of $\Phi(t)$ implies that, for any choice of the
constants $\gamma_j$, with $\sum_{j=1}^n \gamma_j = 1$
 \be\label{gen1}
H(f_1*f_2*\cdots *f_n) \ge \sum_{j=1}^n \gamma_j H(u_j(t)) - \frac d2 \sum_{j=1}^n \gamma_j \log \gamma_j.
 \ee
Inequality \fer{gen1} generalizes to $n$ functions the result of
Lemma \ref{loc2}. Shannon's entropy power inequality then follows by
maximizing the right-hand side of \fer{gen1} over the sequence
$\gamma_j$.

\section{Conclusions}

In this paper we studied the monotonicity properties of various
functionals related to convolutions of powers of solutions to the
heat equation. This monotonicity is at the basis of a  new proof of
many well-known inequalities in sharp form, which are viewed in our
picture as consequence of a unique well understandable physical
principle, in the form of time monotonicity of a Lyapunov
functional. Partial results of this strategy were presented in
\cite{Tos1,Tos2,Tos3}.

This idea has been applied here to prove classical Young's
inequality and its converse, Brascamp--Lieb type inequalities,
Babenko's inequality and Pr\'ekopa--Leindler inequality. In
addition, a new direct proof of Shannon's entropy power inequality
is shown to follow by the same argument.

Unlike similar results obtained in recent years (cf. \cite{Ba,BC,
BH, B2, BB,B1,BBC,B1, Bor,Car}), we were inspired by some relatively
old papers by people working on information theory \cite{Bla,Sta}
and kinetic theory of rarefied gases \cite{McK}, mainly connected
with classical Shannon's entropy and its monotonicity properties.

  %%%%%%%%%%%%%%%%%%%%%%%%%%%%%%%%%%%%%%%%%%%%%%%%%%%%%%%%%%%%%%%%%%%%%%%

% ------------------------------------------------------------------------

% ------------------------------------------------------------------------
\end{document}